\documentclass[11pt]{article}

\usepackage[utf8]{inputenc}
\usepackage[english]{babel}

\usepackage{hyperref}
\usepackage{xcolor}
\usepackage{amsmath,amssymb}
\usepackage{enumerate}
\usepackage{amsthm}
\usepackage[all,2cell]{xy}
\usepackage{mathrsfs}
\usepackage{mathtools}
\usepackage{tikz}

\usepackage{geometry}
\def\baselinestretch{1.1}
\newtheorem{thm}{Theorem}[section]
\newtheorem{dfn}[thm]{Definition}
\newtheorem{prop}[thm]{Proposition}

\newtheorem{lem}[thm]{Lemma}
\newtheorem{exmpl}[thm]{Example}
\newtheorem{obs}[thm]{Remark}

\def\beq{\begin{equation}}
\def\eeq{\end{equation}}
\def\bea{\begin{eqnarray}}
\def\eea{\end{eqnarray}}
\def\beann{\begin{eqnarray*}}
\def\eeann{\end{eqnarray*}}
\def\ben{\begin{enumerate}}
\def\een{\end{enumerate}}
\def\bit{\begin{itemize}}
\def\eit{\end{itemize}}

\def\ds{\displaystyle}

\def\derpar#1#2{\frac{\partial{#1}}{\partial{#2}}}
\newcommand\restr[2]{{
  \left.\kern-\nulldelimiterspace 
  #1 
  \right|_{#2} 
}}

\newcommand{\R}{\mathbb{R}}

\renewcommand{\d}{\mathrm{d}}
\renewcommand{\L}{\mathcal{L}}
\renewcommand{\H}{\mathcal{H}}
\newcommand{\vf}{\mathfrak{X}}
\newcommand{\df}{\Omega}
\newcommand{\Cinfty}{\mathscr{C}^\infty}

\newcommand{\Tan}{\mathrm{T}}
\newcommand{\inn}{i}
\newcommand{\Lie}{\mathscr{L}}
\newcommand{\U}{\mathcal{U}}
\newcommand{\X}{\mathfrak{X}}
\newcommand{\Reeb}{\mathcal{R}}

\title{\Large\sffamily
New contributions to the Hamiltonian and Lagrangian contact formalisms\\
for dissipative mechanical systems
and their symmetries}
\author{\sffamily
$^a$Jordi Gaset, 
$^b$Xavier Gr\`acia, 
$^b$Miguel C. Mu\~noz-Lecanda,
$^b$Xavier Rivas and 
$^b$Narciso Rom\'an-Roy%
\thanks{emails:
jordi.gaset@uab.cat,
xavier.gracia@upc.edu,
miguel.carlos.munoz@upc.edu,
xavier.rivas@upc.edu,
narciso.roman@upc.edu}
\\[1ex]
\normalsize\itshape\sffamily 
$^a$Dept.\ of Physics,
Universitat Aut\`onoma de Barcelona.
\normalsize\itshape\sffamily 
Bellaterra, Catalonia, Spain
\\[0.1ex]
\normalsize\itshape\sffamily 
$^b$Dept.\ of Mathematics,
Universitat Polit\`ecnica de Catalunya.
\normalsize\itshape\sffamily 
Barcelona, Catalonia, Spain
}
\date{\sffamily 
9 June 2020}

\begin{document}
\maketitle

\kern -12mm
\begin{abstract}
\noindent
We provide new insights into the contact Hamiltonian and Lagrangian formulations
of dissipative mechanical systems.
In particular, we state a new form of the contact dynamical equations, 
and we review two recently presented Lagrangian formalisms,
studying their equivalence.
We define several kinds of symmetries
for contact dynamical systems,
as well as the notion of dissipation laws,
prove a dissipation theorem 
and give a way to construct conserved quantities.
Some well-known examples of dissipative systems 
are discussed.
\end{abstract}

\noindent\textbf{Keywords:}
contact manifold, dissipative system, Lagrangian formalism, Hamiltonian formalism, symmetry, dissipation law.

\noindent\textbf{MSC\,2010 codes:}
 37J55, 53D10, 70H05, 70H33, 70G45, 70G65, 70H03.

\medskip
\setcounter{tocdepth}{2}
{
\def\baselinestretch{0.97}
\small
\def\addvspace#1{\vskip 1pt}
\parskip 0pt plus 0.1mm
\tableofcontents
}

\newpage
\section{Introduction}

Many theories in modern physics can be formulated using the tools of differential geometry.
For instance, 
the natural geometric framework for autonomous Hamiltonian mechanical systems is symplectic geometry 
\cite{abraham1978,arn78,lib87}, 
whereas its nonautonomous counterpart can be nicely described using cosymplectic geometry
\cite{CLM91,chi94,LS-2016}.
Geometric approaches to field theory make use of
multisymplectic and
$k$-symplectic geometry,
among others
\cite{Carinena1991,LeSaVi2016,RRSV2011}.
All these methods are developed to model systems
of variational type without dissipation or damping, 
both in the Lagrangian and Hamiltonian formalisms.

In recent years there has been a growing interest in studying 
a geometric framework to describe dissipative systems 
\cite{CG-2019,Galley-2013,MPR-2018,Ra-2006}.
An especially useful tool is contact geometry \cite{BHD-2016,BGG-2017,CNY-2013,Geiges-2008}.
Contact geometry has been used to describe several types of physical systems
in thermodynamics, quantum mechanics, circuit theory, control theory, etc.\
(see for instance
\cite{Bravetti-2019,Goto-2016,KA-2013,RMS-2017}).
In recent articles 
\cite{GGMRR-2019,GGMRR-2020}
a generalization of contact geometry has been used
to describe field theories with dissipation.

In this paper we are primarily interested in mechanics.
Contact geometry has been applied to give
a Hamiltonian-type description of mechanical systems with dissipation
\cite{bravetti2017,BCT-2017,LL-2018,LIU2018,Vi-2018},
although in recent papers the Lagrangian formalism has been also considered
\cite{Ciaglia2018,DeLeon2019}.
The model introduced in \cite{DeLeon2019} is more general 
and it is inspired in the cosymplectic formulation
of non-autonomous mechanical systems;
meanwhile the other one \cite{Ciaglia2018} 
describes a more particular, but very usual, 
type of Lagrangian systems with dissipation.

Our aim is to enlarge these lines of research, 
giving new insights into the contact formulation of these systems,
both in the Hamiltonian and the Lagrangian formalisms.
While carefully reviewing these formulations,
we present a new form of the contact Hamiltonian and Lagrangian equations.
We also expand the formulation introduced in 
\cite{Ciaglia2018} 
and compare it with the one given in  
\cite{DeLeon2019}, 
and we show that they are equivalent
for the Lagrangians with dissipation studied in 
\cite{Ciaglia2018}.
It is interesting to point out that 
the first Lagrangian formulation of contact systems
was introduced by Herglotz, from a variational perspective, 
around 1930
\cite{He-1930,Her-1985}
(see also \cite{GeGu2002}).

We also introduce and study symmetries and dissipated quantities in a geometrical way.
We define different kinds of symmetries for contact Hamiltonian and Lagrangian systems 
(the so-called dynamical and contact),
depending on which structure is preserved, 
and establish the relations among them. 
We have also compared the characteristics of symmetries 
for symplectic and contact dynamical systems, 
showing that there are significant differences
between them. 
Other relevant results are 
the statements of the so-called ``dissipation theorems'', 
which are analogous to the conservation theorems of conservative systems,
and the association of dissipated and conserved quantities to these symmetries.

Some well-known physical examples 
are discussed using this geometric framework.
In particular, we propose contact Lagrangian functions 
that describe
the damped harmonic oscillator, the motion in a gravitational field with friction, and the parachute equation.

The paper is organized as follows.
Section \ref{prel} is devoted to review several preliminary concepts on
contact geometry and contact Hamiltonian systems,
including a new setting of the dynamical equations.
In Section \ref{section-contact} we study the
Lagrangian formalism for dissipative systems in detail,
and we compare the formulations given in
\cite{Ciaglia2018} and~\cite{DeLeon2019}.
In Section~\ref{sims} we analyze the concepts of 
symmetry and dissipation law for dissipative systems,
prove a dissipation theorem
and show how to obtain conserved quantities. 
Finally, Section \ref{examples} is devoted to the examples.

Throughout the paper all the manifolds and mappings are assumed to be smooth. 
Sum over crossed repeated indices is understood.

\section{Preliminaries}
\label{prel}

In this section we present some geometric structures 
that will be used to describe the Lagrangian formalism 
of dissipative dynamical systems
(see, for instance,
\cite{%
bravetti2017,
BCT-2017,
Geiges-2008,
LL-2018} 
for details).

\subsection{Contact geometry and contact Hamiltonian systems}

\begin{dfn}\label{dfn-contact-manifold}
    Let $M$ be a $(2n+1)$-dimensional  manifold. 
    A \textbf{contact form} in $M$ is a differential $1$-form
    $\eta\in\df^1(M)$ such that $\eta\wedge(\d\eta)^{\wedge n}$
    is a volume form in $M$.
    Then, $(M,\eta)$ is said to be a \textbf{contact manifold}.
\end{dfn}

As a consequence of the condition that
$\eta\wedge(\d\eta)^{\wedge n}$ is a volume form we have a decomposition of 
$\Tan M$, induced by $\eta$, in the form 
$\Tan M= \ker\d\eta\oplus\ker\eta\equiv\mathcal{D}^{\rm R}\oplus\mathcal{D}^{\rm C}$. 

\begin{prop}
Let $(M,\eta)$ be a contact manifold. 
Then there exists a unique vector field $\Reeb\in\X(M)$, 
which is called \textbf{Reeb vector field}, 
such that
\begin{equation}\label{eq-Reeb}
    \begin{cases}
        \inn(\Reeb)\d\eta = 0,\\
        \inn(\Reeb)\eta = 1.
    \end{cases}
\end{equation}
\end{prop}

This vector field generates the distribution ${\cal D}^{\rm R}$, 
which is called the {\sl Reeb distribution}.

\begin{prop}\label{prop-adapted-coord}
On a contact manifold there exist local coordinates
$(x^I;s)$
such that
$$
    \Reeb = \frac{\partial}{\partial s}
    \,,\quad
    \eta = \d s - f_I(x) \,\d x^I
    \,,
$$
where $f_I(x)$ are functions depending only on the~$x^I$.
(These are the so-called 
$\textsl{adapted coordinates}$
of the contact structure.)
\end{prop}
\begin{proof}
Let $(x^I,s)$, $I=1,\ldots,n$, local coordinates in $M$ rectifying the vector field $\Reeb$, that is $\Reeb=\partial/\partial s$ in the domain of these coordinates.

Then $\eta=a\,\d s-f_I(x,s)\d x^I$. The conditions defining $\Reeb$ imply that $a=1$ and $\partial f_I/\partial s=0$, hence the result is proved.
\end{proof}

Nevertheless, one can go further and prove the existence of Darboux-type coordinates:

\begin{thm}[Darboux theorem for contact manifolds]
Let $(M,\eta)$ be a contact manifold. 
Then around each point $p\in M$ there exist a chart 
$(\U; q^i, p_i, s)$ with $1\leq i\leq n$ such that
\begin{equation*}
	\restr{\eta}{\U} = \d s - p_i\,\d q^i \  .
\end{equation*}
These are the so-called \textbf{Darboux} or \textbf{canonical coordinates} of the contact manifold $(M,\eta)$.
\end{thm}

In Darboux coordinates, the Reeb vector field is
$\displaystyle\restr{\Reeb}{\U} = \frac{\partial}{\partial s}$.

\begin{exmpl}
\label{example-canmodel}
\rm
(Canonical contact structure).
The manifold
$M=\Tan^*Q\times \R$
has a canonical contact structure defined by the 1-form
$\eta= \d s- \theta$,
where $s$
is the cartesian coordinate of $\R$,
and 
$\theta$
is the pull-back of the canonical 1-form of $\Tan^*Q$
with respect to the projection
$M \to \Tan^*Q$.
Using coordinates $q^i$ on~$Q$ and natural coordinates
$(q^i,p_i)$ on $\Tan^*Q$,
the local expressions of the contact form is
$\eta=\d s-p_i \,\d q^i$;
from which 
$\d \eta=\d q^i\wedge\d p_i$,
and the Reeb vector field is
$\displaystyle\Reeb= \frac{\partial}{\partial s}$.
\end{exmpl}

\begin{exmpl}
\label{example2}
\rm
(Contactification of a symplectic manifold).
Let $(P,\omega)$ be a symplectic manifold such that $\omega=-\d\theta$,
and consider $M =P\times\R$.
Denoting by $s$ the cartesian coordinate of~$\R$,
and representing also by $\theta$ the pull-back of~$\theta$
to the product, we consider the 1-form
$\eta=\d s-\theta$ on~$M$.
Then $(M,\eta)$ is a contact manifold which is called the
\textbf{contactified} of~$P$.
Notice that the previous example, \ref{example-canmodel}, 
is just the contactification of $\Tan^*Q$ with its canonical symplectic structure.
\end{exmpl}

Finally, given a contact manifold $(M,\eta)$, 
we have the following $\Cinfty(M)$-module isomorphism
\begin{equation*}
\begin{array}{rccl}
    \flat\colon & \X(M) & \longrightarrow & \df^1(M) \\
    & X & \longmapsto & \inn(X)\d\eta+(\inn(X)\eta)\eta
\end{array}
\end{equation*}


\begin{thm}
\label{teo-hameqs}
    If $(M,\eta)$ is a contact manifold, for every $\H\in\Cinfty(M)$,
     there exists a unique vector field $X_\H\in\X(M)$ such that
    \begin{equation}\label{hamilton-contact-eqs}
        \begin{cases}
            \inn(X_\H)\d\eta=\d\H-(\Lie_{\Reeb}\H)\eta\\
            \inn(X_\H)\eta=-\H \ .
        \end{cases}
    \end{equation}
    Then, the integral curves of $X_\H$, ${\bf c}\colon I\subset\R\to M$,
    are the solutions to the equations
    \begin{equation}\label{hamilton-contactc-curves-eqs}
        \begin{cases}
            \inn({\bf c}')\d\eta=\left(\d\H-(\Lie_{\Reeb}\H)\eta\right)\circ{\bf c}\\
            \inn({\bf c}')\eta=-\H\circ{\bf c} \ ,
        \end{cases}
    \end{equation}
    where ${\bf c}'\colon I\subset\R\to \Tan M$ is the canonical lift of the curve
    ${\bf c}$ to the tangent bundle $\Tan M$.
\end{thm}

\begin{dfn}
The vector field $X_\H$ is the
\textbf{contact Hamiltonian vector field} associated to $\H$ and the equations 
\eqref{hamilton-contact-eqs} and \eqref{hamilton-contactc-curves-eqs}
are the \textbf{contact Hamiltonian equations} 
for this vector field and its integral curves, respectively.

The triple $(M,\eta,\H)$ is a \textbf{contact Hamiltonian system}.
\end{dfn}

As a consequence of the definition of $X_\H$ we have
the following relation,
which expresses the
{\sl dissipation of the Hamiltonian}:
\beq
\Lie_{X_\H}\H =
-(\Lie_{\Reeb}\H)\H \:.
\label{eq:disipenerg}
\eeq
The proof is immediate:
$$
\Lie_{X_\H}\H=-\Lie_{X_\H}\inn(X_\H)\eta=
-\inn(X_\H)\Lie_{X_\H}\eta=
\inn(X_\H)((\Lie_\Reeb\H)\eta)=-(\Lie_{\Reeb}\H)\H \ .
$$

\begin{prop}
\label{prop-assertions}
Let $(M,\eta,\H)$ be a contact Hamiltonian system.
The following assertions are equivalent:
\begin{enumerate}
\item 
$X_\H$ is a contact Hamiltonian vector field (i.e., a solution to equations \eqref{hamilton-contact-eqs}).
\item
$X_\H$ is a solution to the equations
\begin{equation*}
\begin{cases}
    \Lie_{X_\H} \eta= -(\Lie_\Reeb\H) \, \eta \:,
    \\
    \inn(X_\H)\eta = -\H \:,
\end{cases}
\end{equation*}
\item
$X_\H$ satisfies that
\begin{equation*}
\flat(X_\H) =
\d \H - (\Lie_\Reeb\H + \H) \eta
\:.
\end{equation*}
\end{enumerate}
\end{prop}

In the next section we give another way of stating the equations for a contact Hamiltonian vector field.

Taking Darboux coordinates $(q^i,p_i,s)$, 
the contact Hamiltonian vector field is
$$ 
X_\H = \frac{\partial\H}{\partial p_i}\frac{\partial}{\partial q^i} - 
\left(\frac{\partial\H}{\partial q^i} + 
p_i\frac{\partial\H}{\partial s}\right)\frac{\partial}{\partial p_i} + 
\left(p_i\frac{\partial\H}{\partial p_i} - \H\right)\frac{\partial}{\partial s}\ ; 
$$
and its integral curves $\gamma(t) = (q^i(t), p_i(t), s(t))$
are solutions to the dissipative Hamilton equations  \eqref{hamilton-contactc-curves-eqs}
which read as
\begin{equation}\label{diss-ham-eqs}
\begin{dcases}
    \dot q^i = \frac{\partial\H}{\partial p_i}\ ,\\
    \dot p_i = -\left(\frac{\partial\H}{\partial q^i} + p_i\frac{\partial\H}{\partial s}\right)\ ,\\
    \dot s = p_i\frac{\partial\H}{\partial p_i} - \H\ .
\end{dcases}
\end{equation}

\begin{obs}{\rm
The case in which some of the conditions stated in Definition 
\ref{dfn-contact-manifold} do not hold
has been recently analyzed in \cite{DeLeon2019}, 
where the definitions and properties of
{\sl precontact structures} and {\sl precontact manifolds} 
are introduced and studied.
In particular, for this kind of manifolds, 
Reeb vector fields are not uniquely determined.
In these cases $(M,\eta,\H)$ is called a 
{\sl precontact Hamiltonian system}.
}\end{obs}

\subsection{Equivalent form of the contact Hamiltonian equations}

In this section we present a new way of writing equations \eqref{hamilton-contact-eqs} 
without making use of the Reeb vector field $\Reeb$. 
This can be useful in the case of singular systems because, as we have pointed out, 
in a precontact manifold we do not have a uniquely determined Reeb vector field.

\begin{prop}
\label{eqset}
Let $(M, \eta, \H)$ be a contact Hamiltonian system and consider the open set $U=\{p\in M;\H(p)\neq 0\}$.
Let $\Omega$ be the 2-form defined by $\Omega = -\H\,\d\eta + \d\H\wedge\eta$ on~$U$. 
A vector field $X\in\X(U)$ is the contact Hamiltonian vector field if, and only if,
    \begin{equation}\label{hamilton-eqs-no-reeb}
        \begin{cases}
            \inn(X)\Omega = 0\,,\\
            \inn(X)\eta = -\H\,.
        \end{cases}
    \end{equation}
\end{prop}
\begin{proof}
Suppose that $X$ satisfies equations \eqref{hamilton-eqs-no-reeb}. Then,
\begin{equation*}
    0 = \inn(X)\Omega = -\H\,\inn(X)\d\eta +(\inn(X)\d\H)\eta + \H\,\d\H\ ,
\end{equation*}
and hence
\begin{equation}\label{no-reeb-1}
    \H\,\inn(X)\d\eta =(\inn(X)\d\H)\eta + \H\,\d\H\ .
\end{equation}
Contracting with the Reeb vector field,
\begin{equation*}
    0 = \H\,\inn(\Reeb)\inn(X)\d\eta =
    (\inn(X)\d\H) \inn(\Reeb)\eta + \H\,\inn(\Reeb)\d\H\ ,
\end{equation*}
and $\inn(X)\d\H = -\H\,\inn(\Reeb)\d\H$. 
Using this in equation \eqref{no-reeb-1}, we get
\begin{equation*}
    \H\inn(X)\d\eta=\H(\d\H -(\inn(\Reeb)\d\H)\eta)=
    \H(\d\H -(\Lie_\Reeb\H)\eta)\ .
\end{equation*}
and hence $\inn(X)\d\eta=\d\H -(\Lie_\Reeb \H)\eta$.
    
Now suppose that $X$ satisfies equations  \eqref{hamilton-contact-eqs}. 
Then:
\beann
 \inn(X)\Omega&=&\inn(X)(-\H\d\eta+\d\H\wedge\eta)=
    -\H\inn(X)\d\eta+(\inn(X)\d\H)\eta+\H\d\H
    \\
 &=& \H(\Lie_\Reeb\H)\eta+(\Lie_X\H)\eta=
    (\H(\Lie_\Reeb\H)+(\Lie_X\H))\eta=0 \ ,
    \eeann
hence $\inn(X)\Omega=0$, using the dissipation of the Hamiltonian $\H$ \eqref{eq:disipenerg}.    
\end{proof}

\begin{obs}{\rm
Let $p\in M$ and suppose that $\H(p)=0$. Then the second equation in both groups \eqref{hamilton-contact-eqs} 
and \eqref{hamilton-eqs-no-reeb} gives $\inn(X_p)\eta=0$, 
hence $X_p\in \ker\eta_p$. The remaining equation in \eqref{hamilton-contact-eqs} 
is $\inn(X_p)\d\eta=\d\H(p) -(\Lie_{\Reeb_p}\H)\eta_p$,
and the corresponding one in \eqref{hamilton-eqs-no-reeb} is $\inn(X_p)\Omega=(\Lie_{X_p}H)\eta_p=0$. 
These two equations are not equivalent. In fact the first one implies the second, 
using the dissipation of the Hamiltonian, but not on the contrary.
}
\end{obs}

Finally, bearing in mind Theorem \ref{teo-hameqs}, we can state:

\begin{prop}
\label{eqset2}
On the open set $U=\{p\in M;\H\neq 0\}$, 
a path ${\bf c}\colon I\subset\R\to M$ 
is an integral curve of the 
contact Hamiltonian vector field $X_\H$
if, and only if, it is a solution to
$$
\begin{cases}
\inn({\bf c}')\Omega = 0 \:,
\\
\inn({\bf c}')\eta = - \H\circ{\bf c} \:.
\end{cases}
$$
\end{prop}

\section{Lagrangian formalism for dissipative dynamical systems}
\label{section-contact}

\subsection{Lagrangian phase space and geometric structures}
\label{kLagForm}

If $Q$ is an $n$-dimensional manifold,
consider the product manifold 
$\Tan Q\times\R$ 
with canonical projections
$$
s\colon \Tan Q\times\R\to\R 
\quad , \quad
\tau_1\colon \Tan Q\times\R\to\Tan Q
 \quad , \quad
\tau_0\colon \Tan Q\times\R\to Q\times\R\ .
$$
Notice that 
$\tau_1$ and $\tau_0$ 
are the projection maps of two vector bundle structures. 
We will usually have the second one in mind;
indeed, 
with this structure
$\Tan Q \times \R$
is the pull-back of the tangent bundle $\Tan Q$ 
with respect the map $Q \times \R \to Q$.
Natural coordinates in $\Tan Q\times\R$ will be denoted $(q^i,v^i, s)$.

In order to develop a contact Lagrangian formalism 
we need to extend the usual geometric structures of Lagrangian mechanics
to the Lagrangian phase space.
Notice that, as a product manifold, we can write
$\Tan(\Tan Q \times \R) =
(\Tan(\Tan Q)) \times \R) \oplus (\Tan Q \times \Tan\R)
$,
so any operation that can act on tangent vectors to $\Tan Q$
can act on tangent vectors to $\Tan Q \times \R$.
In particular,
the vertical endomorphism 
of $\Tan(\Tan Q)$ yields a
\textbf{vertical endomorphism}
${\cal J} \colon \Tan (\Tan Q\times\R) \to \Tan (\Tan Q\times\R)$.
Similarly,
the Liouville vector field on $\Tan Q$
yields a \textbf{Liouville vector field}
$\Delta \in \vf(\Tan Q\times\R)$;
indeed, 
this is the Liouville vector field 
of the vector bundle structure defined by $\tau_0$.

In natural coordinates,
the local expressions of these objects are
$$
{\cal J} =
\frac{\partial}{\partial v^i} \otimes \d q^i
\,,\quad
\Delta = 
v^i\, \frac{\partial}{\partial v^i}
\,.
$$

\begin{dfn}
\label{de652}
Let 
${\bf c} \colon\R \rightarrow Q\times\R$ 
be a path,
with 
${\bf c} = (\mathbf{c}_1,\mathbf{c}_2)$.
The \textbf{prolongation} 
of ${\bf c}$ 
to $\Tan Q\times\R$ 
is the path
$$
{\bf \tilde c} = 
(\mathbf{c}_1',\mathbf{c}_2)
\colon 
\R \longrightarrow \Tan Q \times \R  \,,
$$
where $\mathbf{c}_1'$ is the velocity of~$\mathbf{c}_1$.
The path ${\bf \tilde c}$ is said to be 
\textbf{holonomic}.
\\
A vector field 
$\Gamma \in \X(\Tan Q \times \R)$ 
is said to satisfy the 
\textbf{second-order condition}
(for short:  is a {\sc sode}) when all of its integral curves 
are holonomic. 
\end{dfn}

This definition can be equivalently expressed in terms of the canonical structures defined above:

\begin{prop}
A vector field 
$\Gamma  \in \X(\Tan Q\times\R)$ 
is a \textsc{sode} if,
and only if,
${\cal J} \circ \Gamma = \Delta$.
\end{prop}

In coordinates, if 
${\bf c}(t)=(c^i(t), s(t))$, 
then
$$
{\bf \tilde c}(t) =
\left( c^i(t),\frac{d c^i}{d t}(t), s(t) \right) \:.
$$
The local expression of a {\sc sode} is
\begin{equation}
\label{localsode2}
\Gamma= 
v^i \frac{\partial}{\partial q^i} +
f^i \frac{\partial}{\partial v^i} + 
g\,\frac{\partial}{\partial s}
\:.
\end{equation}
So, in coordinates a {\sc sode} defines a system of
differential equations of the form
$$
\frac{\d^2 q^i}{\d t^2}=f^i(q,\dot q,s) \:, \quad  
\frac{d s}{d t}=g(q,\dot q,s)  \:.
$$

\subsection{Contact Lagrangian systems}
\label{sec-conLagsys}

A general setting of the Lagrangian formalism for dissipative mechanical systems
has been recently developed by 
Le\'on and Lainz-Valc\'azar \cite{DeLeon2019}.
Next we review this formulation and give new additional features.

\begin{dfn}
\label{lagrangean}
A \textbf{Lagrangian function} 
is a function $\L \colon\Tan Q\times\R\to\R$.
\\    
The \textbf{Lagrangian energy}
associated with $\L$ is the function $E_\L := \Delta(\L)-\L\in\Cinfty(\Tan Q\times\R)$.
\\    
The \textbf{Cartan forms}
associated with $\L$ are defined as
\begin{equation}
\label{eq:thetaL}
\theta_\L = 
{}^t{\cal J} \circ \d \L 
\in \Omega^1(\Tan Q\times\R) 
\:,\quad
\omega_\L = 
-\d \theta_\L
\in \Omega^2(\Tan Q\times\R) 
\:.
\end{equation}
The \textbf{contact Lagrangian form} is
$$
\eta_\L=\d s-\theta_\L\in\Omega^1(\Tan Q\times\R)
\:;
$$
it satisfies that $\d\eta_\L=\omega_\L$.
\\
The couple $(\Tan Q\times\R,\L)$ is a \textbf{contact Lagrangian system}.
\end{dfn}

Taking natural coordinates $(q^i, v^i, s)$ in $\Tan Q\times\R$, 
the form $\eta_\L$ is written as
\begin{equation}
\label{eq:etaL}
\eta_\L = \d s - \frac{\partial\L}{\partial v^i} \,\d q^i \:,
\end{equation}
and consequently
\begin{equation*}
\d\eta_\L = 
-\frac{\partial^2\L}{\partial s\partial v^i}\d s\wedge\d q^i 
-\frac{\partial^2\L}{\partial q^j\partial v^i}\d q^j\wedge\d q^i 
-\frac{\partial^2\L}{\partial v^j\partial v^i}\d v^j\wedge\d q^i\ .
\end{equation*}

The next structure to be defined is the Legendre map.
Before, remember that,
given a (not necessarily linear) bundle map $f \colon E \to F$
between two vector bundles over a manifold~$B$,
the fibre derivative of~$f$ is the map
$\mathcal{F}f \colon 
E \to \mathrm{Hom}(E,F) \approx F \otimes E^*$
obtained by
restricting $f$ to the fibres, 
$f_b \colon E_b \to F_b$,
and computing the usual derivative
of a map between two vector spaces:
$\mathcal{F}f(e_b) = \mathrm{D} f_b(e_b)$
---see \cite{Gracia2000} for a detailed account.
This applies in particular 
when the second vector bundle is trivial of rank~1,
that is,
for a function
$f \colon E \to \R$;
then
$\mathcal{F}f \colon E \to E^*$.
This map also has a fibre derivative
$\mathcal{F}^2 f \colon E \to E^* \otimes E^*$,
which can be called the fibre hessian of~$f$;
for every $e_b \in E$,
$\mathcal{F}^2 f(e_b)$ can be considered as a
symmetric bilinear form on $E_b$.
It is easy to check that
$\mathcal{F}f$ is a local diffeomorphism 
at a point $e \in E$
if and only if
the Hessian $\mathcal{F}^2f(e)$ is non-degenerate.

\begin{dfn}
Given a Lagrangian 
$\L \colon \Tan Q\times\R \to \R$, 
its \textbf{Legendre map}
is the fibre derivative of~$\L$,
considered as a function on the vector bundle
$\tau_0 \colon \Tan Q\times\R \to Q \times \R$;
that is, the map
${\cal F}\L \colon \Tan Q \times\R \to \Tan^*Q \times \R$ 
given by
$$
{\cal F}\L (v_q,s) = \left( \strut {\cal F}\L(\cdot,s) (v_q),s \right)
\,,
$$
where $\L(\cdot,s)$ is the Lagrangian with $s$ freezed.
\end{dfn}

\begin{obs}\rm
The Cartan forms can also be defined as
$$
\theta_\L={\cal FL}^{\;*}(\pi_1^*\theta)
\,,\quad
\omega_\L={\cal FL}^{\;*}(\pi_1^*\omega)
\,.
$$
\end{obs}

\begin{prop}
\label{Prop-regLag}
For a Lagrangian function $\L$ the following conditions are equivalent:
\begin{enumerate}
\item
The Legendre map
${\cal FL}$ is a local diffeomorphism.
\item
The fibre Hesian
$
{\cal F}^2\L \colon
\Tan Q\times\R \longrightarrow (\Tan^*Q\times\R)\otimes (\Tan^*Q\times\R)
$
of~$\L$ is everywhere nondegenerate.
(The tensor product is of vector bundles over $Q \times \R$.)
\item
$(\Tan Q\times\R,\eta_\L)$
is a contact manifold.
\end{enumerate}
\end{prop}

The proof of this result can be easily done 
using natural coordinates, where
$$
{\cal FL}:(q^i,v^i, s)  \longrightarrow  \left(q^i,
\frac{\displaystyle\partial\L}{\displaystyle\partial v^i}, s
\right)\,,
$$
$$
{\cal F}^2 \L(q^i,v^i,s) = (q^i,W_{ij},s)
\,,\quad
\hbox{with}\ 
W_{ij} = 
\left( \frac{\partial^2\L}{\partial v^i\partial v^j}\right)
\,.
$$
Then the conditions in the proposition 
mean that the matrix 
$W= (W_{ij})$ 
is everywhere nonsingular.

\begin{dfn}
A Lagrangian function $\L$ is said to be \textbf{regular} if the equivalent
conditions in Proposition \ref{Prop-regLag} hold.
Otherwise $\L$ is called a \textbf{singular} Lagrangian.
In particular, 
$\L$ is said to be \textbf{hyperregular} 
if ${\cal FL}$ is a global diffeomorphism.
\end{dfn}

\begin{obs}{\rm
As a result of the preceding definitions and results,
every \emph{regular} contact Lagrangian system 
has associated the contact Hamiltonian system
$(\Tan Q\times\R, \eta_\L, E_\L)$.
}
\end{obs}

Given a regular contact Lagrangian system $(\Tan Q\times\R,\L)$,
from \eqref{eq-Reeb} we have that
the \textbf{Reeb vector field} 
$\Reeb_\L\in\X(\Tan Q\times\R)$ 
for this system is uniquely determined by the relations
\begin{equation}\label{eq:reeb}
\begin{cases}
    \inn(\Reeb_\L)\d\eta_\L=0\ ,\\
    \inn(\Reeb_\L)\eta_\L=1 \ .
\end{cases}
\end{equation}
Its local expression is
\begin{equation}
\label{coorReeb}
\Reeb_\L=\frac{\partial}{\partial s}-W^{ji}
\frac{\partial^2\L}{\partial s \partial v^j}\,\frac{\partial}{\partial v^i} \ ,
\end{equation}
where $(W^{ij})$ is the inverse of the Hessian matrix,
namely 
$W^{ij} W_{jk} = \delta^i_{\,k}$.

Notice that the Reeb vector field does not appear
in the simplest form $\partial/\partial s$.
This is due to the fact that the natural coordinates in $\Tan Q \times \R$
are \emph{not} adapted coordinates for $\eta_\L$.

\subsection{The contact Euler--Lagrange equations}

\begin{dfn}
\label{def-lageqs}
Let $(\Tan Q\times\R,\L)$ be a regular contact Lagrangian system.
\\
The \textbf{contact Euler--Lagrange equations} for a holonomic curve
${\bf\tilde c}\colon I\subset\R \to\Tan Q\times\R$ are
\begin{equation}
\begin{cases}
\inn({\bf\tilde c}')\d\eta_\L = 
\Big(\d E_\L - (\Lie_{\Reeb_\L}E_\L)\eta_\L\Big)\circ{\bf\tilde c} 
\:,\\
\inn({\bf\tilde c}')\eta_\L = - E_\L\circ{\bf\tilde c} 
\:,
\end{cases}
\label{hec}
\end{equation}
where ${\bf\tilde c}'\colon I\subset\R\to\Tan(\Tan Q\times\R)$ denotes the
canonical lifting of ${\bf\tilde c}$ to $\Tan(\Tan Q\times\R)$.
\\
The \textbf{contact Lagrangian equations} for a vector field $X_\L\in\X(\Tan Q\times\R)$ are 
\begin{equation}
\label{eq-E-L-contact1}
\begin{cases}
    \inn(X_\L)\d \eta_\L=\d E_\L-(\Lie_{\Reeb_\L}E_\L)\eta_\L\ ,
\\
    \inn(X_\L)\eta_\L=-E_\L \ .
\end{cases}
\end{equation}
A vector field which is a solution to these equations is called a
\textbf{contact Lagrangian vector field}
(it is a contact Hamiltonian vector field for the function $E_\L$).
\end{dfn}

\begin{obs}{\rm
Now, taking into account Propositions \ref{eqset} and \ref{eqset2},
in the open set $U=\{p\in M;\,E_\L(p)\neq 0\}$,
the above equations can be stated equivalently as
\beq\label{hamilton-eqs-no-reeb-lag2}
 \begin{cases}
\inn({\bf \tilde c}')\Omega_\L = 0 \ ,\\
\inn({\bf \tilde c}')\eta_\L = - E_\L\circ{\bf\tilde c} \ ,
\end{cases}
\eeq
and
    \begin{equation}\label{hamilton-eqs-no-reeb-lag}
        \begin{cases}
           \inn(X_\L)\Omega_\L = 0\ ,\\
            \inn(X_\L)\eta_\L = -E_\L\ ,
        \end{cases}
    \end{equation}
where $\Omega_{\cal L}=-E_{\cal L}\,\d\eta_{\cal L} + \d E_{\cal L}\wedge\eta_{\cal L}$.
}\end{obs}

In natural coordinates, for a holonomic curve
${\bf\tilde c}(t)=(q^i(t),\dot q^i(t),s(t))$,
equations \eqref{hec} are
\bea
\label{ELeqs2}
\dot s&=&\L \ ,
 \\
\label{ELeqs3}
\displaystyle\frac{\partial^2\L}{\partial v^j \partial v^i}\,
\ddot q^j +\displaystyle\frac{\partial^2\L}{\partial q^j \partial v^i} \,\dot q^j  +
\displaystyle \frac{\partial^2\L}{\partial s \partial v^i}\, \dot s
-\displaystyle\frac{\partial\L}{ \partial q^i}=
\frac{d}{dt}\left(\derpar{\L}{v^i}\right)-
\displaystyle\frac{\partial\L}{\partial q^i}&=&
\displaystyle\frac{\partial\L}{\partial s}\displaystyle\frac{\partial\L}{\partial v^i} \ ,
\eea
(which coincide with the so-called {\sl generalized Euler-Lagrange equations} stated in  \cite{He-1930});
meanwhile, for a vector field
$X_\L=\displaystyle f^i\,\frac{\partial}{\partial q^i}+F^i\,\frac{\partial}{\partial v^i} +
g\,\frac{\partial}{\partial s}$,
equations \eqref{eq-E-L-contact1} are
\bea
\displaystyle
\left( f^j-v^j \right)
\frac{\partial^2\L}{\partial v^j \partial s}
&=&0 \,,
\label{A-E-L-eqs2}
\\
\left( f^j-v^j \right)
\frac{\partial^2\L}{\partial v^i \partial v^j}
&=&0
\label{A-E-L-eqs1} \,,
\\
\left( f^j-v^j \right)
\frac{\partial^2\L}{\partial q^i \partial v^j}
+\frac{\partial\L}{\partial q^i}
-\frac{\partial^2\L}{\partial s \partial v^i}g
-\frac{\partial^2\L}{\partial q^j \partial v^i}f^j
-\frac{\partial^2\L}{\partial v^j \partial v^i}F^j
+\frac{\partial\L}{\partial s}
\frac{\partial\L}{\partial v^i}
&=& 0\,,
\label{A-E-L-eqs3}
\\
\L + 
\frac{\partial\L}{\partial v^i} (f^i-v^i) -g 
&=& 0\,;
\label{A-E-L-eqs4}
\eea
For these computations it is useful to use the relation
\beq
\label{ReebLag}
\Lie_{\Reeb_\L} E_\L = - \derpar{\L}{s}
\,,
\eeq
which is easily proved in coordinates.

\begin{prop}
\label{ELeq-teor}
If $\L$ is a regular Lagrangian, then $X_\L$ is a {\sc sode} and
the equations \eqref{A-E-L-eqs4} and \eqref{A-E-L-eqs3} become
\bea
\label{ELeqs0}
g&=&\L \ ,
\\
\label{ELeqs1}
\displaystyle\frac{\partial^2\L}{\partial v^j \partial v^i}\,
F^j +\displaystyle\frac{\partial^2\L}{\partial q^j \partial v^i} \,v^j  +
\displaystyle \frac{\partial^2\L}{\partial s \partial v^i}\,\L
-\displaystyle\frac{\partial\L}{ \partial q^i}&=&
\displaystyle\frac{\partial\L}{\partial s}\displaystyle\frac{\partial\L}{\partial v^i} \ ,
\eea
which, for the integral curves of $X_\L$, are the Euler--Lagrange equations
\eqref{ELeqs2} and \eqref{ELeqs3}.

This {\sc sode} $X_\L\equiv\Gamma_\L$ is called the \textbf{Euler--Lagrange vector field} 
associated with the Lagrangian function $\L$.
\end{prop}
\proof
It follows from the coordinate expressions.
If $\L$ is a regular Lagrangian, equations \eqref{A-E-L-eqs1}
lead to $v^i=f^i$, which
are the {\sc sode} condition for the vector field $X_\L$.
Then, \eqref{A-E-L-eqs2} holds identically, and
\eqref{A-E-L-eqs4} and \eqref{A-E-L-eqs3} give the equations \eqref{ELeqs0} and \eqref{ELeqs1}
or, equivalently, for the integral curves of $X_\L$, the Euler--Lagrange equations \eqref{ELeqs2} and \eqref{ELeqs3}.
\qed

In this way, the local expression of this Euler--Lagrange vector field is
\beq
\label{sode-coor}
\Gamma_\L=
\L\,\frac{\partial}{\partial s}
+v^i\,\frac{\partial}{\partial q^i}
+W^{ik}
\left(
\frac{\partial\L}{ \partial q^k} 
- \frac{\partial^2\L}{\partial q^j \partial v^k} \,v^j
- \L\frac{\partial^2\L}{\partial s \partial v^k} 
+\frac{\partial\L}{\partial s}
 \frac{\partial\L}{\partial v^k} 
\right)
\frac{\partial}{\partial v^i} \,.
\eeq

\begin{obs}{\rm
It is interesting to point out how, in the Lagrangian formalism of dissipative systems,
the expression in coordinates \eqref{ELeqs2} of
the second Lagrangian equation \eqref{eq-E-L-contact1} relates the variation 
of the ``dissipation coordinate'' $s$ to the Lagrangian function and, from here, 
we can identify this coordinate with the Lagrangian action, $\displaystyle s=\int{\cal L}\,dt$.
}
\end{obs}

\begin{obs}{\rm
If $\L$ is singular, although $(\Tan Q\times\R,\eta_\L)$ is not
a contact manifold, but a pre-contact one, and hence the Reeb vector field
is not uniquely defined, it can be proved that
the Lagrangian equations \eqref{eq-E-L-contact1} 
are independent on the Reeb vector field used  \cite{DeLeon2019}.
Alternativelly, Proposition \ref{eqset} holds also in this case and,
hence, the Reeb-independent equations \eqref{hamilton-eqs-no-reeb-lag}
can be used instead.
In any case, solutions to the Lagrangian equations
are not necessarily {\sc sode} and,
in order to obtain the Euler--Lagrange equations \eqref{ELeqs1}
(or \eqref{ELeqs3}), 
the condition ${\cal J}(X_\L)=\Delta$ must be added to the above Lagrangian equations.
Furthermore, these equations are not necessarily compatible everywhere on $\Tan Q\times\R$ 
and a suitable {\sl constraint algorithm} must be implemented in order to find 
a {\sl final constraint submanifold} $S_f\hookrightarrow\Tan Q\times\R$
(if it exists) where there are {\sc sode} vector fields $X_{\cal L}\in\X(\Tan Q\times\R)$,
tangent to $S_f$, which are (not necessarily unique) solutions to the above equations on $S_f$.
All these problems have been studied in detail in  \cite{DeLeon2019}.
}
\end{obs}

\subsection{The canonical Hamiltonian formalism for contact Lagrangian systems}
\label{Legmap}

In the (hyper)regular case we have a diffeomorphism between $(\Tan Q\times\R,\eta_\L)$ and
$(\Tan^*Q\times\R,\eta)$,
where ${\cal FL}^{\;*}\eta=\eta_\L$.
Furthermore, there exists (maybe locally) a function $\H\in\Cinfty(\Tan^* Q\times\R)$ 
such that ${\cal FL}^{\;*}\H=E_\L$; then we have the
contact Hamiltonian system $(\Tan^*Q\times\R,\eta,\H)$,
for which ${\cal FL}_*{\Reeb}_\L={\Reeb}$.
Then, if $X_\H\in\X(\Tan^*Q\times\R)$ is the
contact Hamiltonian vector field associated with $\H$,
we have that ${\cal FL}_*\Gamma_\L=X_\H$.

For singular Lagrangians, following \cite{got79} we define:

\begin{dfn}
A singular Lagrangian $\L$ is said to be \textbf{almost-regular} if 
$\mathcal{P}:= {\cal FL}(\Tan Q)$ is
a closed submanifold of $\Tan^*Q\times\R$, 
the Legendre map ${\cal FL}$ is a submersion onto its image, and
the fibres ${\cal FL}^{-1}({\cal FL}(v_q, s))$, 
for every $(v_q, s)\in \Tan Q\times\R$, are
connected submanifolds of $\Tan Q\times\R$.
\end{dfn}

In these cases, we have $({\cal P},\eta_{\cal P})$, where
$\eta_{\cal P}=j_{\cal P}^*\eta\in\Omega^1({\cal P})$,
and $j_{\cal P}\colon{\cal P}\hookrightarrow\Tan^*Q\times\R$
is the natural embedding.
Furthermore, the Lagrangian energy function $E_\L$ is 
${\cal FL}$-projectable; 
i.e., there is a unique $\H_{\cal P}\in\Cinfty({\cal P})$
such that $E_\L={\cal FL}_o^*\,\H_{\cal P}$,
where ${\cal FL}_o\colon\Tan Q\times\R\to{\cal P}$
is the restriction of  ${\cal FL}$ to ${\cal P}$,
defined by  ${\cal FL}=j_{\cal P}\circ{\cal FL}_o$.
Then, there exists a Hamiltonian formalism associated with the original Lagrangian system, 
which is developed on the submanifold ${\cal P}$,
and the contact Hamiltonian equations for
$X_{\H_{\cal P}}\in\X({\cal P})$ are \eqref{hamilton-contact-eqs}
adapted to this situation or, equivalently,
  \begin{equation}\label{HeqsP}
        \begin{cases}
        \inn(X_{\H_{\cal P}})\Omega_{\cal P} = 0\ ,\\
        \inn(X_{\H_{\cal P}})\eta_{\cal P} = -\H_{\cal P}\ ,
        \end{cases}
    \end{equation}
where $\Omega_{\cal P}=-\H_{\cal P}\,\d\eta_{\cal P} + \d \H_{\cal P}\wedge\eta_{\cal P}$.
As in the Lagrangian formalism, these equations are not necessarily compatible everywhere on ${\cal P}$ 
and a suitable {\sl constraint algorithm} must be implemented in order to find 
a {\sl final constraint submanifold} $P_f\hookrightarrow{\cal P}$
(if it exists) where there are vector fields $X_{\H_{\cal P}}\in\X({\cal P})$,
tangent to $P_f$, which are (not necessarily unique) solutions to \eqref{HeqsP} on $P_f$.
(See \cite{DeLeon2019} for a detailed analysis on all these topics).

\subsection{Lagrangians with holonomic dissipation term}
\label{dislag}

In a recent paper by Ciaglia, Cruz, and Marmo
\cite{Ciaglia2018}
a Lagrangian description for some systems with dissipation was given
using a modification of Lagrangian formalism 
inspired by the contact Hamiltonian formalism.
In this section we will see that this description
coincides with the general formalism studied in 
Section \ref{sec-conLagsys}
when applied to a particular class of contact Lagrangians.

\begin{dfn}
A \textbf{Lagrangian with holonomic dissipation term} in $\Tan Q\times\R$ 
is a function $\L=L+\phi\in\Cinfty(\Tan Q\times\R)$, 
where $L=\tau_1^{\,*}L_o$, for a Lagrangian function
$L_o\in\Cinfty(\Tan Q)$, and $\phi=\tau_0^{\,*}\phi_o$, for $\phi_o\in\Cinfty(Q\times\R)$.
\label{simple}
\end{dfn}

In coordinates, $\L(q^i,v^i,s)=L(q^i,v^i)+\phi(q^i,s)$.
Observe that this implies that the momenta defined by the Legendre transformation are independent of the coordinate~$s$. 
In addition, for these Lagrangians the conditions
$\displaystyle\frac{\partial^2\L}{\partial v^i\partial s}=0$ hold. 
This motivates the name given in the definition.

\begin{obs}{\rm
The Lagrangian formalism presented in \cite{Ciaglia2018} 
is a little less general than the one we examine here, 
since there only the case $\phi=\phi(s)$ is considered.
}\end{obs}

\begin{prop}
Let $\L = L+\phi$ be a Lagrangian with holonomic dissipation term.
Then
its Cartan 1-form, contact form, energy and Reeb vector field
as a contact Lagrangian
can be computed as
$$
\theta_\L = \theta_L,\quad 
\eta_\L = \d s - \theta_L,\quad  
E_\L = E_L - \phi,\quad
\Reeb_\L = \frac{\partial}{\partial s}
\,,
$$
where
$\theta_L$ is the Cartan 1-form of~$L$ considered (via  pull-back)
as a 1-form on $\Tan Q \times \R$,
and 
$E_L$ is the energy of~$L$ as a function in the same way on $\Tan Q \times \R$.

The Legendre map of~$\L$,
${\cal F\L} \colon \Tan Q\times\R \to \Tan^*Q\times\R$,
can be expressed as
${\cal F\L}={\cal F}L \times {\rm Id}_\R$,
where ${\cal F}L$ is the Legendre map of~$L$.
The Hessians are related by
${\cal F}^2 \L(v_q,s) = {\cal F}^2 L(v_q)$.
Moreover,
$\L$ is regular if, and only if, $L$ is regular.
\end{prop}
\begin{proof}
It is immediate in coordinates.
In particular, the assertion about the Legendre map
is a consequence of the fact that
$\partial \L / \partial v^i = \partial L / \partial v^i$.
In a similar way the relation between the Hessians 
can be expressed in coordinates as
$\displaystyle 
\frac{\partial^2\L}{\partial v^i\partial v^j}=
\frac{\partial^2L}{\partial v^i\partial v^j}$.
This shows that $\L$ is regular if, and only if, $L$ is regular.
\end{proof}

Obviously $\L$ is hyperregular if, and only if, $L$ also is.
This means that 
the Legendre map ${\cal F}\L$ is a diffeomorphism,
and the canonical Hamiltonian formalism 
for the Lagrangian with nonholonomic dissipation term
can be formulated as in Section~\ref{Legmap}.

For the contact Lagrangian system 
$(\Tan Q\times\R,\eta_\L,E_\L)$ 
the dynamical equations 
for vector fields are
\begin{equation*}
\begin{cases}
 \inn(X_\L)\d\eta_\L=
 \d E_\L-(\Lie_{\Reeb_\L}E_\L) \eta_\L\ ,\\
   \inn(X_\L)\eta_\L= -E_\L \,.
\end{cases}
\end{equation*}
In coordinates, writing
$\displaystyle 
X_\L=
g\frac{\partial}{\partial s}+f^i\frac{\partial}{\partial q^i}+F^i\frac{\partial}{\partial v^i}$,
the second Lagrangian equation 
for $X_\L$ reads, in coordinates,
\beq
\L+\frac{\partial L}{\partial v^i} (f^i-v^i) -g=0 \ ,
\label{preg}
\eeq
and this  is the equation \eqref{A-E-L-eqs4} for $\L=L+\phi$.
The first Lagrangian equation gives
\begin{equation}
(f^i-v^i)\frac{\partial^2{L}}{\partial v^j\partial v^i}=0 \ ,
\label{presode}
\end{equation}
and
\beq
\left(\frac{\partial^2{L}}{\partial q^i\partial v^j}
-\frac{\partial^2{L}}{\partial q^j\partial v^i}\right)\,f^j+
\frac{\partial^2{L}}{\partial q^i\partial v^j}\,v^j-
\frac{\partial^2{L}}{\partial v^j\partial v^i}\,F^j =
-\frac{\partial L}{\partial q^i} - \frac{\partial\phi}{\partial q^i}-
\frac{\partial\phi}{\partial s} \frac{\partial L}{\partial v^i}\ , 
\label{preEL}
\eeq
which is equation \eqref{A-E-L-eqs3} for  $\L$.
Observe that equations \eqref{A-E-L-eqs2} are identities, since
$\displaystyle \frac{\partial^2\L}{\partial v^j \partial s}=0$.

Finally, as in Proposition \ref{ELeq-teor},
if $\L$ is a regular Lagrangian
then Eqn.~(\ref{presode}) implies that $f^i = v^i$,
that is,
$X_\L$ is a {\sc sode}, 
and the equations of motion become
\bea
\label{1st}
g&=&\L \ , \\
\label{2nd}
\frac{\partial^2{L}}{\partial v^j\partial v^i}\,\ddot{q}^j+
\frac{\partial^2{L}}{\partial q^j\partial v^i}\,\dot{q}^j-
\frac{\partial L}{\partial q^i}=
 \frac{d}{dt}\left(\derpar{L}{v^i}\right)-\displaystyle\frac{\partial L}{ \partial q^i}&=&
\frac{\partial\phi}{\partial q^i}+\frac{\partial\phi}{\partial s} \frac{\partial L}{\partial v^i} \,.
\eea
These are the expression in coordinates of the contact Euler--Lagrange equations.

\section{Symmetries of contact Hamiltonian and Lagrangian systems}
\label{sims}

\subsection{Symmetries for contact Hamiltonian systems}

For a dynamical system, there are different concepts of symmetry 
which depend on which structure they preserve. 
Thus, one can consider the transformations that preserve the geometric structures of the system,
or those that preserve its solutions \cite{Gracia2002,R2019}.
Next we discuss these subjects for contact systems.
(See also \cite{DeLeon2019b} for another complementary approach on these topics).

Let $(M,\eta,\H)$ be a contact Hamiltonian system with Reeb vector field $\Reeb$,
and $X_\H$ the contact Hamiltonian vector field for this system;
that is, the solution to the Hamilton equations \eqref{hamilton-contact-eqs}.

\begin{dfn}\label{dfn:sym}
    A \textbf{dynamical symmetry} of a contact Hamiltonian system is a diffeomorphism 
    $\Phi\colon M\longrightarrow M$ such that $\Phi_*X_\H=X_\H$ 
    (it maps solutions into solutions).
    
    An \textbf{infinitesimal dynamical symmetry} of a contact Hamiltonian system is a vector field 
    $Y\in \X(M)$ whose local flow is a dynamical symmetry; that is,
    $\Lie_YX_\H=[Y,X_\H]=0$.
\end{dfn}

There are other kinds of symmetries that let
the geometric structures invariant.
They are the following:

\begin{dfn}
\label{defsym}
A \textbf{contact symmetry} of a contact Hamiltonian system is a diffeomorphism 
$\Phi\colon M\rightarrow M$ such that
$$
\Phi^*\eta=\eta
\ ,\quad 
\Phi_*\H=\H \,.
$$
An \textbf{infinitesimal contact symmetry} of a contact Hamiltonian system is a vector field 
$Y\in \X(M)$ whose local flow is a contact symmetry; that is,
$$
\Lie_Y\eta=0
\ ,\quad 
\Lie_Y\H=0 \,.
$$
\end{dfn}

Furthermore we have:

\begin{prop}
    Every (infinitesimal) contact symmetry preserves the Reeb vector field; that is, $\Phi^*\Reeb=\Reeb$
    (or $[Y,\Reeb]=0$).
\end{prop}
\begin{proof}
    We obtain that
    \beann
\inn(\Phi^{-1}_*\Reeb)(\Phi^*\d\eta)&=&
\Phi^*(\inn(\Reeb)\d\eta)=0 \ , \\
\inn(\Phi^{-1}_*\Reeb)(\Phi^*\eta)&=&
\Phi^*(\inn(\Reeb)\eta)=1 \ ,
    \eeann
and, as $\Phi^*\eta=\eta$ and the Reeb vector fields are unique, 
from these equalities we get $\Phi^{-1}_*\Reeb=\Reeb$.
    The proof for the infinitesimal case is immediate from the definition.
\end{proof}

Finally, as a consequence of all of this 
we obtain the relation among contact symmetries and dynamical symmetries:

\begin{prop}
(Infinitesimal) contact symmetries are (infinitesimal) dynamical symmetries.
\end{prop}
\begin{proof}
    If $X_\H$ is the contact Lagrangian vector field, then 
    \beann
        \inn(\Phi_*X_\H)\d\eta&=&
        \inn(\Phi_*X_\H)(\Phi^*\d\eta)=
        \Phi^*(\inn(X_\H)\d\eta)
        \\ &=&
        \Phi^*(\d\H-(\Lie_\Reeb \H)\eta)=\d\H-(\Lie_\Reeb \H )\eta ,
        \\
        \inn(\Phi_*X_\H)\eta&=&
        \inn(\Phi_*X_\H)(\Phi^*\eta)=
        \Phi^*(\inn(X_\H)\eta)=
        \Phi^*(-\H)=-\H \ .
    \eeann
    The proof for the infinitesimal case is immediate from the definition.
\end{proof}


\subsection{Dissipated and conserved quantities of contact Hamiltonian systems}

Associated with symmetries of contact Hamiltonian systems are
the concepts of {\sl dissipated} and {\sl conserved quantities}:

\begin{dfn}\label{dfn:diss-quan}
    A \textbf{dissipated quantity} of a contact Hamiltonian system is a function 
    $F\in\Cinfty(M)$ satisfying that 
    \beq{}\label{eq:dissipation}
    \Lie_{X_\H}F=-(\Lie_\Reeb \H )\,F \ .
    \eeq{}
\end{dfn}

For contact Hamiltonian systems, symmetries are associated with dissipated quantities as follows:


\begin{thm} 
\label{th:dissipation}
{\rm (Dissipation theorem).} 
Let $Y$ be a vector field on~$M$.
If $Y$ is an infinitesimal dynamical symmetry 
($[Y,X_H]=0$),
then the function
$F=-\inn(Y)\eta$
is a dissipated quantity.
\end{thm}
\begin{proof}
This is a consequence of
\begin{align*}
\Lie_{X_\H} F &=
-\Lie_{X_\H} \inn(Y) \eta =
- \inn(Y) \Lie_{X_\H}\eta - \inn(\Lie_{X_H}Y) \eta =
\\
&=
(\Lie_\Reeb \H) \inn(Y)\eta + \inn([Y,X_H]) \eta =
 -(\Lie_\Reeb \H) F + \inn([Y,X_H]) \eta =
 -(\Lie_\Reeb \H) F
 \,,
\end{align*}
where we have applied the second assertion of Proposition \ref{prop-assertions}.
\end{proof}

\begin{obs}{\rm
The last equality shows that, indeed,
$[Y,X_\H] \in \mathop{\mathrm{Ker}} \eta$
is a necessary and sufficient condition for $F$ 
to be a dissipated quantity.
This has been noticed in \cite{DeLeon2019b}
while this paper was under review.
Nevertheless, it should be noted that such transformations are not dynamical symmetries in the sense of Definition \ref{dfn:sym}, 
since they do not transform solutions into solutions, in general.}	
\end{obs}


In particular, as it was established in \eqref{eq:disipenerg},
the Hamiltonian vector field $X_\H$ is
trivially a symmetry and its dissipated quantity is the energy, 
$F=-\inn(X_\H)\eta=\H$;
that is shown in the following:

\begin{thm}{\rm (Energy dissipation theorem).}
    $\Lie_{X_\H}\H=-(\Lie_\Reeb\H)\H$.
    \label{disipe2}
\end{thm}

\begin{obs}{\rm
    Observe that these are ``non conservation theorems''.
    As we are dealing with dissipative systems,
    dynamical symmetries are not associated with conserved quantities, 
    but with dissipative quantities, and then these theorems
    account for the non-conservation of these quantities associated with the symmetries.
    In particular, the energy is not a conserved quantity,
    as it was early commented in Section \ref{prel}.
}
\end{obs}


\begin{dfn}\label{dfn:con-quan}
    A \textbf{conserved quantity} is a function 
    $G\colon M\longrightarrow \R$ satisfying that 
    $$
    \Lie_{X_\H}G=0 \ .
    $$
\end{dfn}

Every dissipated quantity changes with the same rate ($-\Reeb(\H)$), 
which suggests that the quotient of two dissipated quantities should be a conserved quantity. 
Indeed:

\begin{prop}
\label{prop:disscon} 
\begin{enumerate}
    \item If $F_1$ and $F_2$ are dissipated quantities and $F_2\neq0$, then $F_1/F_2$ is a conserved quantity.
    \item If $F$ is a dissipated quantity and $G$ is a conserved quantity, then $FG$ is a dissipated quantity.
\end{enumerate}
\end{prop}
\begin{proof}
    In fact, we have
    \begin{align*}
    \Lie_{X_\H}(F_1/F_2)&=
    F_2^{-1}\Lie_{X_\H}{F_1}-F_1F_2^{-2}\Lie_{X_\H}{F_2}=-F_2^{-1}(\Lie_\Reeb \H )F_1+F_1F_2^{-2}(\Lie_\Reeb\H)F_2=0
     \ .
     \\
     \Lie_{X_\H}(FG)&=G\Lie_{X_\H}F+F\Lie_{X_\H}G=-(\Lie_\Reeb \H) FG.
    \end{align*}
\end{proof}

\begin{obs}{\rm
If $\H\neq 0$, it is possible to assign a conserved quantity to an infinitesimal dynamical symmetry~$Y$. 
Indeed, from Theorem \ref{th:dissipation} and Proposition \ref{prop:disscon}, 
the function $-i(Y)\eta/\H$ is a conserved quantity. 
}
\end{obs}

Finally, contact symmetries can be used to generate new dissipated quantities 
from a given dissipated quantity.
In fact, as a straightforward consequence of definitions 
\ref{dfn:diss-quan} and \ref{dfn:sym} we obtain the following:

\begin{prop}
    If $\Phi\colon M\rightarrow M$ is a contact symmetry and 
    $F\colon M\rightarrow \mathbb{R}$ is a dissipated quantity, then so is $\Phi^*F$.
\end{prop}
\begin{proof}
In fact, we have
$$
\Lie_{X_\H}(\Phi^*F)=
\Phi^*\Lie_{\Phi_*X_\H}F=
\Phi^*\Lie_{X_\H}F=
\Phi^*(-\Lie_\Reeb \H )F=
-(\Lie_\Reeb \H)(\Phi^*F).
$$
The proof for the infinitesimal case is immediate from the definition.
\end{proof}

\subsection{Symmetries for contact canonical Hamiltonian systems}

Let $M=\Tan^*Q\times\mathbb{R}$ 
be canonical contact structure with the contact form 
$\eta=\d s-p_i\d q^i$
(see the example \ref{example-canmodel}).

Remember that, if $\varphi\colon Q\to Q$ is a diffeomorphism,
we can construct the diffeomorphism 
$\Phi:=(\Tan^*\varphi, {\rm {\rm Id}_{\mathbb R}}) \colon
\Tan^*Q\times\R
\longrightarrow 
\Tan^*Q\times\R$,
where $\Tan^*\varphi\colon\Tan^*Q\to\Tan^*Q$
is the canonical lifting of $\varphi$ to $\Tan^*Q$.
Then $\Phi$ is said to be the {\sl canonical lifting} of $\varphi$ to $\Tan^*Q\times\R$.
Any transformation $\Phi$ of this kind is called a {\sl natural transformation} of $\Tan^*Q\times\R$.

In the same way, given a vector field $Z\in \X(Q)$
we can define its {\sl complete lifting}
to $\Tan^*Q\times\R$ as the vector field
$Y\in\X(\Tan^*Q\times\R)$
whose local flow is the canonical lifting of 
the local flow of $Z$ to $\Tan^*Q\times\R$; 
that is, the vector field $Y=Z^{C*}$,
where $Z^{C*}$ denotes the complete lifting of $Z$ to $\Tan^*Q$,
identified in a natural way as a vector field in $\Tan^*Q\times\R$.
Any infinitesimal  transformation $Y$ of this kind is called a {\sl natural infinitesimal transformation} of $\Tan^*Q\times\R$.

It is well-known that the canonical forms $\theta_o$ and $\omega_o=-\d\theta_o$
in $\Tan^*Q$ are invariant under the action
of canonical liftings of diffeomorphisms and vector fields from $Q$ to $\Tan^*Q$.
Then, taking into account the definition
of the contact form $\eta$ in $\Tan^*Q\times\R$,
as an immediate consequence of the above considerations we have:

\begin{prop}
If
$\Phi\in{\rm Diff}(\Tan^*Q\times\R)$ (resp. $Y\in\vf(\Tan^*Q\times\R)$) is a canonical lifting to $\Tan^*Q\times\R$ of a diffeomorphism $\varphi\in{\rm Diff}(Q)$
(resp. of $Z\in\vf(Q)$), then
\begin{enumerate}
    \item
$\Phi^*\eta=\eta$ (resp $\Lie_Y\eta=0$).
    \item
If, in addition,
$\Phi^*\H=\H$ (resp. $\Lie_Y\H=0$),
then it is a  (infinitesimal) contact symmetry.
\end{enumerate}
\end{prop}

In particular, we have the following:

\begin{thm}{\rm (Momentum dissipation theorem).} 
If $\displaystyle\frac{\partial \H}{\partial q^i}=0$, 
then $\displaystyle\frac{\partial}{\partial q^i}$ 
is an infinitesimal contact symmetry, 
and its associated dissipated quantity  
is the corresponding momentum $\displaystyle p_i$;
that is,
$\Lie_{X_\H}p_i=-(\Lie_\Reeb\H)p_i$.
\label{disipe3}
\end{thm}
\begin{proof}
A simple computation in local coordinates shows that 
$\displaystyle
\Lie\Big(\frac{\partial}{\partial q^i}\Big)\eta=0$ 
and 
$\displaystyle
\Lie\Big(\frac{\partial}{\partial q^i}\Big)\H=0$.
Therefore it is a contact symmetry and, in particular, a dynamical symmetry. 
The other results are a consequence of the dissipation theorem. 
\end{proof}

\subsection{Symmetries for contact Lagrangian systems}

Consider a regular contact Lagrangian system $(\Tan Q\times\R,\L)$, 
with Reeb vector field $\Reeb_\L$,
and let $X_\L$ be the contact Euler--Lagrange vector field for this system;
that is, the solution to the Lagrangian equations \eqref{eq-E-L-contact1}.

All we have said about symmetries and dissipated quantities for
contact Hamiltonian systems holds when it is applied
to the contact system $(\Tan Q\times\R,\eta_\L,E_\L)$.
Thus we have the same definitions for dynamical and contact symmetries
and the dissipation theorem states that
$-\inn(Y)\eta_\L$ is a dissipated quantity,
for every infinitesimal dynamical symmetry $Y$.
In particular, the energy dissipation theorem says that
$$
\Lie_{X_\L}E_\L=-(\Lie_{\Reeb_\L}E_\L)E_\L\ .
$$

In the same way as for $\Tan^*Q\times\R$,
if $\varphi\colon Q\to Q$ is a diffeomorphism,
we can construct the diffeomorphism 
$\Phi:=(\Tan\varphi, {\rm {\rm Id}_{\mathbb R}}) \colon
\Tan Q\times\R
\longrightarrow 
\Tan Q\times\R$,
where $\Tan\varphi\colon \Tan Q\to\Tan Q$
is the canonical lifting of $\varphi$ to $\Tan Q$.
Then $\Phi$ is said to be the {\sl canonical lifting} of $\varphi$ to $\Tan Q\times\R$.
Any transformation $\Phi$ of this kind
is called a \emph{natural transformation} of $\Tan Q\times\R$.

Moreover, given a vector field $Z\in \X(\Tan Q)$
we can define its {\sl complete lifting}
to $\Tan Q\times\R$ as the vector field
$Y\in\X(\Tan Q\times\R)$
whose local flow is the canonical lifting of 
the local flow of $Z$ to $\Tan Q\times\R$; 
that is, the vector field $Y=Z^C$,
where $Z^C$ denotes the complete lifting of $Z$ to $\Tan Q$,
identified in a natural way as a vector field in $\Tan Q\times\R$.
Any infinitesimal transformation $Y$ of this kind
is called a \emph{natural infinitesimal transformation} 
of $\Tan Q\times\R$.

It is well-known that the vertical endomorphism $J$ and the Liouville vector field $\Delta_o$
in $\Tan Q$ are invariant under the action
of canonical liftings of diffeomorphisms and vector fields from $Q$ to $\Tan Q$.
Then, taking into account the definitions
of the canonical endomorphism ${\cal J}$
and the Liouville vector field $\Delta$ in $\Tan Q\times\R$,
it can be proved that canonical liftings of diffeomorphisms and vector fields
from $Q$ to $\Tan Q\times\R$ preserve these canonical structures
as well as the Reeb vector field $\Reeb$.

Therefore, as an immediate consequence, 
we obtain a relation between 
Lagrangian-preserving natural transformations 
and contact symmetries:

\begin{prop}
If $\Phi\in{\rm Diff}(\Tan Q\times\R)$ (resp. $Y\in\vf(\Tan Q\times\R)$) is a canonical lifting
to $\Tan Q\times\R$ of a diffeomorphism $\varphi\in{\rm Diff}(Q)$
(resp.\ of a vector field $Z\in\vf(Q)$)
that leaves the Lagrangian invariant, then
it is a (infinitesimal) contact symmetry, {\it i.e.},
$$
    \Phi^*\eta_\L=\eta_\L \:,\
    \Phi^*E_\L=E_\L
    \qquad 
    ({\rm resp.}\ \Lie_Y\eta_\L=0 \:,\
    \Lie_YE_\L=0 \:)\:.
$$
As a consequence, it is a (infinitesimal) dynamical symmetry.
\end{prop}

As a consequence we have a result similar to  the momentum dissipation theorem \ref{disipe3}: if $\displaystyle\frac{\partial \L}{\partial q^i}=0$, 
then $\displaystyle\frac{\partial}{\partial q^i}$ 
is an infinitesimal contact symmetry and it associated dissipated quantity 
is the momentum 
$\displaystyle\frac{\partial \L}{\partial v^i}$;
that is,
$$
\Lie_{X_\L}\left(\frac{\partial \L}{\partial v^i}\right)=
-(\Lie_{\Reeb_\L}E_\L)\frac{\partial \L}{\partial v^i}=
\derpar{\L}{s}\frac{\partial \L}{\partial v^i}\ .
$$

\begin{obs}\rm
A similar problem is considered in \cite{GeGu2002}, where the used dissipation factor is
$\displaystyle \derpar{\L}{s}$ which, as we have seen in
Eq.~(\ref{ReebLag}), is the same that we have obtained.
\end{obs}

\subsection{Symmetries of a contactified system}

The dissipation theorem (Theorem \ref{th:dissipation})
yields a dissipated quantity
from a dynamical symmetry~$Y$
with no additional hypotheses.
This is in contrast to Noether symmetries,
where the generator of the symmetry has to satisfy
some additional hypotheses
in order to yield a conserved quantity.
We want to understand this different behaviour.

So, let us consider a Hamiltonian system 
$(P,\omega,H_\circ)$
on an exact symplectic manifold~$P$, 
with symplectic form
$\omega = -\d \theta$,
and Hamiltonian function $H_\circ \colon P \to \R$.
Its associated Hamiltonian vector field
$X_\circ \in \vf(P)$
is defined by
$\inn_{X_\circ} \omega = \d H_\circ$.

The contactified of $(P,\omega)$
is the contact manifold $(M,\eta)$,
where
$M = P \times \R$
is endowed with the contact form
$\eta = \d s - \theta$;
here $s$ is the Cartesian coordinate of~$\R$,
and we use $\theta$ for the pull-back of the 1-form to the product 
(see example \ref{example2}).
The pull-back of $H_\circ$ to~$M$
defines a contact Hamiltonian function 
$H = H_\circ$ on~$M$.
The corresponding contact Hamiltonian vector field 
can be written
$\ds
X = X_\circ + \ell \derpar{}{s}
$,
where 
$X_\circ$ is the Hamiltonian vector field of~$H_\circ$
as a vector field on the product manifold, and 
$
\ell = \langle \theta,X_\circ \rangle - H_\circ
$.

Now, let 
$\ds Y_\circ \in \vf(P)$
be a vector field,
and construct 
$\ds
Y = Y_\circ + b \derpar{}{s}
$,
with $b \colon P \to \R$ a function.

\begin{lem}
The vector field $Y$ is a dynamical symmetry of~$X$
($\Lie_Y X = 0$)
if, and only if,
$Y_\circ$ is a dynamical symmetry of~$X_\circ$ 
($\Lie_{Y_\circ} X_\circ = 0$)
and
$\Lie_{Y} \ell = \Lie_{X} b$.
\end{lem}
\begin{proof}
The proof is a consequence of a direct calculation:
$$
[Y,X] =
[Y_\circ,X_\circ] + 
\left( \Lie_{Y} \ell - \Lie_{X} b \right) \derpar{}{s}
\,.
\vadjust{\kern -8mm}
$$
\end{proof}

Now let us consider the quantity
$$
G =-\inn(Y)\eta=-\inn(Y_o)\eta - b
\,.
$$
We have
\begin{align*}
\Lie_X G 
&=
-\inn(Y)\Lie_{X}\eta-\inn([X,Y])\eta
\\
&=
(\Lie_{\Reeb}H)\inn(Y)\eta
-\inn(Y)(\Lie_{X}\eta+(\Lie_\Reeb H)\eta)+\inn([Y,X])\eta
\\
&=
(\Lie_{\Reeb}H)\inn(Y)\eta
-\inn(Y)(\Lie_{X}\eta+(\Lie_\Reeb H)\eta)
+\inn([Y_\circ,X_\circ])\eta
+\left( \Lie_{Y} \ell - \Lie_{X} b \right) 
\,.
\end{align*}
Let us analyze the vanishing of these summands:
the first one is zero because $H=H_\circ$ does not depend on~$s$,
the second one also is because $X$ is the contact Hamiltonian vector field,
and, according to the lemma, 
the third and forth ones vanish 
if 
$\Lie_Y X = 0$.
So we conclude that,
if $Y$ is a dynamical symmetry of~$X$,
$G$ is a conserved quantity for the contact dynamics~$X$
(conserved rather than dissipated, 
since $H$ does not depend on~$s$).
Now, notice that
$$
\Lie_X G = 
\Lie_{X_\circ} \inn(Y_\circ)\theta 
- \Lie_X b
\,.
$$
So we arrive at the following interpretation:
when $Y_0$ is a dynamical symmetry of $X_\circ$,
the function 
$G_\circ = \inn(Y)\theta$
is not necessarily a conserved quantity of~$X_\circ$
(the hypothesis of Noether's theorem are not necessarily satisfied).
However,
in the contactified Hamiltonian system,
with $Y$ a dynamical symmetry projectable to~$Y_\circ$,
$$
\Lie_X G_\circ =
\Lie_X (G+b) =
\Lie_X b =
\Lie_Y \ell
\,,
$$
and the time-derivative of $G_\circ$ is compensated by the time-derivative of~$b$.

When is $G_\circ$ conserved by $X_\circ$?
When 
$\Lie_{Y_\circ} \ell = 0$.
This happens, for instance, when 
$Y_\circ$ is an exact Noether symmetry, i.e., when
$\Lie_{Y_\circ} \theta = 0$
and
$\Lie_{Y_\circ} H_\circ = 0$,
since
$$
\Lie_{Y_\circ} \ell = 
\Lie_{Y_\circ}
  \left( \inn(X_\circ)\theta - H_\circ \right) =
0
\,.
$$

\section{Examples}
\label{examples}

\subsection{The damped harmonic oscillator}

The Lagrangian description of the one dimensional harmonic oscillator is given by $Q= \R$ and 
$\displaystyle L=\frac{1}{2}mv^{2} -\frac{1}{2}m\omega^{2}q^{2}$. The Euler--Lagrange equation gives:
$$
\ddot{q}+\omega^{2}q=0 .
$$
Taking $M=\R\times\R$  and the contact form $\eta=\d s-v\d q$, we can introduce the Lagrangian $\L=L+\phi$, with holonomic dissipation term $\phi=-\gamma s$, and thus
$\displaystyle 
E_\L = \frac{1}{2}mv^2 + \frac{1}{2}m\omega^2q^2 + \gamma s
$. 
We have as dynamical equation:
$$
\ddot{q}=-\omega^{2}q-\gamma\dot{q} 
$$
corresponding to the dynamical equation of the damped harmonic oscillator. 
Therefore, the contact Lagrangian vector field is
$$
\Gamma_{\L}=\L\frac{\partial}{\partial s}+v\frac{\partial}{\partial q}-
\left(\omega^{2}q+\gamma v\right) \frac{\partial}{\partial v}\  .
$$

The dissipation of the energy is given by Theorem~\ref{disipe2}
$$
\Lie_{\Gamma_\L}E_\L=-\gamma E_\L\ .
$$

\subsection{Motion in a constant gravitational field with friction}

Consider the motion of a particle in a vertical plane under the action of constant gravity; then $Q=\R^2$ with coordinates $(x,y)$. 
This motion is described by the Lagrangian $L=\frac{1}{2}mv^2 - mgy$,
where $v^2 = v_x^2 + v_y^2$ in $\Tan Q$ with coordinates $(x,y,v_x,v_y)$.

To introduce air friction, we consider the Lagrangian with holonomic dissipation term $\L = L - \gamma s$ in 
$M=\Tan Q\times\R$, with coordinates $(x,y,v_y,v_x,s)$.
We have
\begin{align*}
    \theta_\L &= mv_x\d x + mv_y\d y\ ,\\
    \eta_\L &= \d s - \theta_\L = \d s - mv_x\d x - mv_y\d y\ ,\\
    \d\eta_\L &= m\d x\wedge \d v_x + m\d y\wedge\d v_y\ ,\\
    E_\L &= \frac{1}{2}mv^2 +mgy + \gamma s\ ,\\
    \Reeb_\L &= \frac{\partial}{\partial s}\ .
\end{align*}

The dynamical equations are
\begin{align}
    \label{E-Lag-1}
    i(X)\d\eta_\L &= \d E_\L - \Reeb_\L(E_\L)\eta_\L\ ,\\
    \label{E-Lag-2}
    i(X)\eta &= -E_\L\ .
\end{align}
Writing $X\in\X(\Tan Q\times\R)$ as
\begin{equation*}
X = 
a\frac{\partial}{\partial s} + b\frac{\partial}{\partial x} + 
c\frac{\partial}{\partial y} +
d\frac{\partial}{\partial v_x} + e\frac{\partial}{\partial v_y}\:,
\end{equation*}
and using that
$$\d E_\L = mv_x\d v_x + mv_y\d v_y + mg\d y + \gamma\d s$$ 
and that
$$\Reeb_\L(E_\L)\eta_\L= \gamma\d s - \gamma mv_x\d x - \gamma mv_y\d y$$,
from equation \eqref{E-Lag-2} we get that 
$a = bv_x + cv_y - E_\L$,
while equation \eqref{E-Lag-1} gives the following conditions:
\begin{equation*}
b = v_x \,,\quad
c = v_y \,,\quad
d = -\gamma v_x \,,\quad
e = -\gamma v_y - g \,.
\end{equation*}
The first two conditions imply that $a = \L$. 
Summing up, the contact Lagrangian vector field is
\begin{equation*}
    \Gamma_\L = \L\frac{\partial}{\partial s} + v_x\frac{\partial}{\partial x} + v_y\frac{\partial}{\partial y} -
    \gamma v_x\frac{\partial}{\partial v_x} - (g + \gamma v_y)\frac{\partial}{\partial v_y}\ .
\end{equation*}
Hence, we get the following system of differential equations:
\begin{equation*}
\begin{cases}
    \ddot{x} + \gamma \dot{x} = 0  \:,\\
    \ddot{y} + \gamma \dot{y} + g = 0  \:,\\
    \dot{s} = \L \:.
\end{cases}
\end{equation*}

As in the previous example, the energy dissipation is given by Theorem \ref{disipe2}: 
$$
\Lie_{\Gamma_\L}E_\L=-\gamma E_\L.
$$ 
We also have that $\displaystyle\frac{\partial \L}{\partial x}=0$, thus 
it is immediate to check that $\displaystyle\frac{\partial }{\partial x}$ is a contact symmetry. 
The associated dissipated quantity is its corresponding momentum: 
$\displaystyle p^x=\frac{\partial \L}{\partial v_x}=mv_x$.
The dissipation of this quantity is given by Theorem \ref{disipe3}
$$
\Lie_Xp^x=-\gamma p^x \ .
$$
Finally, as an application of Proposition \ref{prop:disscon} we have the following conserved quantity
$$
\frac{E_\L}{p^x}=
\frac{\frac{1}{2}mv^2 +mgy + \gamma s}{mv_x} \, .
$$

\subsection{Parachute equation}

Consider the vertical motion of a particle falling in a fluid under the action of constant gravity. If the friction is modelled by the drag equation, the force is proportional to the square of the velocity. This motion can be described as a contact dynamical system in $M=\Tan \R\times\R$, with coordinates $(y,v,s)$, by the Lagrangian function 
\begin{equation*}
\L = 
\frac12mv^2 
-\frac{mg}{2\gamma}\left(e^{2\gamma y}-1\right)
+2\gamma vs \, ,
\end{equation*}
where $\gamma$ is a friction coefficient 
depending on the density of the air, the shape of the object, etc.
Note that this is \emph{not} a Lagrangian with holonomic dissipation term.
We have:
\begin{align*}
\theta_\L &= \left(mv+2\gamma s\right)\d y\ ,\\
\eta_\L &= \d s - \theta_\L = \d s - \left(mv+2\gamma s\right)\d y\ ,\\
\d\eta_\L &= m\,\d y\wedge \d v + 2\gamma\,\d y\wedge\d s\ ,\\
E_\L &= \frac{1}{2}mv^2 +\frac{mg}{2\gamma}\left(e^{2\gamma y}-1\right)\ ,\\
\Reeb_\L &= \frac{\partial}{\partial s}-\frac{2\gamma}{m}\frac{\partial}{\partial v}\ .
\end{align*}

Now we want to write the contact dynamical equations
 \eqref{eq-E-L-contact1}.
Writing the vector field as
\begin{equation*}
    X = a\frac{\partial}{\partial s} + b\frac{\partial}{\partial y} + c\frac{\partial}{\partial v}\in\X(\Tan \R\times\R)\,,
\end{equation*}
and using 
\begin{align*}
\d E_\L &= mv\,\d v+mg\,e^{2\gamma y}\,\d y \ ,
\\
\Reeb_\L(E_\L)\eta_\L&=
-2\gamma v \left(\strut \d s - \left(mv+2\gamma s\right)\d y\right) \ ,
\end{align*}
we derive the following conditions:
\begin{equation*}
\begin{cases}
a = 
(mv+2\gamma)b - E_\L  \,, 
\\
b = 
v  \,,
\\
\displaystyle 
c = 
-g\,e^{2\gamma y} - \frac{2\gamma}{m}a + 2\gamma v^2 + \frac{4\gamma^2}{m}vs \,.
\end{cases}
\end{equation*}
The first two imply that $a = \L$. 
Using this, the third one becomes 
$c=-g + \gamma v^2$.
Summing up we have that the contact Lagrangian vector field is
\begin{equation*}
\Gamma_\L = 
v\frac{\partial}{\partial y}
+ (\gamma v^2-g) \frac{\partial}{\partial v}
+ \L\,\frac{\partial}{\partial s} 
\:.
\end{equation*}
Hence, we get the following system of differential equations
\begin{equation}
\label{eq:parachute}
\begin{cases}
\ddot{y} - \gamma \dot{y}^2 + g = 0\ ,\\
\dot{s} = \L\ .
\end{cases}
\end{equation}
As in the previous example, the energy dissipation is given by Theorem \ref{disipe2},
$$
\Lie_{\Gamma_\L}E_\L= 2\gamma v E_\L \,.
$$ 

\begin{obs}\rm
Eq.~(\ref{eq:parachute}) describes an object falling ($\dot y < 0$).
To describe an object ascending it suffices to
change $\gamma$ for $-\gamma$ in the Lagrangian.
\end{obs}

\section{Conclusions and outlook}

We have studied the geometric formulation
of the Lagrangian and Hamiltonian formalisms of
dissipative mechanical systems using contact geometry.
While contact Hamiltonian systems are well-known,
contact Lagrangian systems have been introduced very recently
\cite{Ciaglia2018,DeLeon2019}. 
Within this study
we have given a new expression of the dynamical equations 
of a contact Hamiltonian system
without the use of the Reeb vector field.
This setting would be especially useful in the case of singular contact systems, 
in which we have a pre-contact structure and then Reeb vector fields are not uniquely determined.

We have reviewed the Lagrangian formalism of contact systems,
which takes place in the bundle $\Tan Q\times\R$.
The most general framework for this formalism is established in \cite{DeLeon2019},
while in \cite{Ciaglia2018} a different formalism is presented
and applied to a particular type of Lagrangians with a dissipation term, 
which is very common in physical applications.
Both formulations have been compared, 
studying their equivalence.

We have analyzed the concept of symmetry for dissipative systems
introducing different kinds of them: for contact Hamiltonian systems in general,
those preserving the dynamics and those preserving the contact structure and the Hamiltonian function,
and, for contact Lagrangian systems, the natural symmetries preserving the Lagrangian function.
The relation among them has been established.

We have introduced the notion of dissipated quantity and, as
the most relevant result, we have proved a {\sl dissipation theorem}, which establishes a relationship between infinitesimal symmetries and dissipated quantities. 
This is analogous to the Noether theorem for conservative systems.
We have also proved that the quotient of two dissipated quantities is a conserved quantity.

Finally, we have also compared the notion of symmetries for symplectic dynamical systems
and dissipative contactified systems, showing that the requirements to be a dynamical symmetry are substantially different.

As some examples we have discussed several common models in physics.
We have proposed contact Lagrangian systems
which lead to the standard dynamical equations 
for the damped harmonic oscillator, 
the motion in a gravitational field with friction, 
and the parachute motion.

Further research on these topics could include 
the study of the non-autonomous case in the Hamiltonian and Lagrangian formalisms,
the characterization of equivalent Lagrangians,
or to pursue the analysis of the geometric structures in the singular case.
On the other hand,
it would be interesting 
to develop a geometric framework 
for Lagrangian field theories with dissipation, 
expanding the results recently obtained in \cite{GGMRR-2019}.

\subsection*{Acknowledgments}
We acknowledge the financial support from the 
Spanish Ministerio de Econom\'{\i}a y Competitividad
project MTM2014--54855--P, 
the Ministerio de Ciencia, Innovaci\'on y Universidades project
PGC2018-098265-B-C33,
and the Secretary of University and Research of the Ministry of Business and Knowledge of
the Catalan Government project
2017--SGR--932.
We thank the referee for his careful reading of the manuscript and his constructive suggestions and comments.

\bibliographystyle{abbrv}
\addcontentsline{toc}{section}{References}
\itemsep 0pt plus 1pt
\small


\begin{thebibliography}{10}

\bibitem{abraham1978}
R.~Abraham and J.~Marsden.
\newblock {\em Foundations of mechanics}.
\newblock Benjamin/Cummings, 1978.

\bibitem{arn78}
V.~I. Arnold.
\newblock {\em Mathematical Methods of Classical Mechanics}.
\newblock Springer-Verlag, 1978.

\bibitem{BHD-2016}
A.~Banyaga and D.~F. Houenou.
\newblock {\em A Brief Introduction to Symplectic and Contact Manifolds}.
\newblock World Scientific, 2016.
(\url{https://doi.org/10.1142/9667}).

\bibitem{bravetti2017}
A.~Bravetti.
\newblock {Contact Hamiltonian dynamics: the concept and its use}.
\newblock {\em Entropy}, {\bf 10}(19):535, 2017.
(\url{https://doi.org/10.3390/e19100535}).

\bibitem{Bravetti-2019}
A.~Bravetti.
\newblock Contact geometry and thermodynamics.
\newblock {\em Int. J. Geom. Methods Mod. Phys.},  {\bf 16}(supp01):1940003, 2019.
(\url{https://doi.org/10.1142/S0219887819400036}).

\bibitem{BCT-2017}
A.~Bravetti, H.~Cruz, and D.~Tapias.
\newblock Contact {H}amiltonian mechanics.
\newblock {\em Ann. Phys.}, {\bf 376}:17 -- 39, 2017.
(\url{https://doi.org/10.1016/j.aop.2016.11.003}).

\bibitem{BGG-2017}
A.~J. {Bruce}, K.~{Grabowska}, and J.~{Grabowski}.
\newblock {Remarks on contact and Jacobi Geometry}.
\newblock {\em Symm. Integ. Geom. Meth. Appl. (SIGMA)}, 
{\bf 13}:059, 2017.
(\url{https://doi.org/10.3842/SIGMA.2017.059}).

\bibitem{CNY-2013}
B.~Cappelletti-Montano, A.~De~Nicola, and I.~Yudin.
\newblock A survey on cosymplectic geometry.
\newblock {\em Rev. Math. Phys.}, {\bf 25}(10):1343002, 2013.
(\url{https://doi.org/10.1142/S0129055X13430022}).

\bibitem{Carinena1991}
J.~Cari{\~{n}}ena, M.~Crampin, and L.~A. Ibort.
\newblock {On the multisymplectic formalism for first order field theories}.
\newblock {\em Diff. Geom. App.}, {\bf 1}(4):345--374, 1991.
(\url{https://doi.org/10.1016/0926-2245(91)90013-Y}).

\bibitem{CG-2019}
J.~Cari{\~{n}}ena and P.~Guha.
\newblock Nonstandard Hamiltonian structures of the Liénard equation
and contact geometry.
\newblock {\it Int. J. Geom. Meth. Mod. Phys.} {\bf 16} (supp 01), 1940001 (2019).
(\url{https://doi.org/10.1142/S0219887819400012}).

\bibitem{CLM91}
D.~Chinea, M.~de~Le\'on, and J.~C. Marrero.
\newblock {Locally conformal cosymplectic manifolds and time-dependent
  Hamiltonian systems}.
\newblock {\em Comment. Math. Univ. Carolin.}, {\bf 32}(2):383--387, 1991.

\bibitem{chi94}
D.~Chinea, M.~de~Le\'on, and J.~C. Marrero.
\newblock The constraint algorithm for time-dependent {L}agrangians.
\newblock {\em J. Math. Phys.}, {\bf 35}(7):3410--3447, 1994.
(\url{https://doi.org/10.1063/1.530476}).

\bibitem{Ciaglia2018}
F.~Ciaglia, H.~Cruz, and G.~Marmo.
\newblock Contact manifolds and dissipation, classical and quantum.
\newblock {\em Ann. Phys.}, {\bf 398}:159 -- 179, 2018.
(\url{https://doi.org/10.1016/j.aop.2018.09.012}).

\bibitem{DeLeon2019}
M.~de~Le{\'{o}}n and M.~Lainz-Valc{\'{a}}zar.
\newblock {Singular Lagrangians and precontact Hamiltonian Systems}.
\newblock {\em Int. J. Geom. Meth. Mod. Phys.},  {\bf 16}(10):1950158, 2019.
(\url{https://doi.org/10.1142/S0219887819501585}).

\bibitem{DeLeon2019b}
M.~de~Le{\'{o}}n and M.~Lainz-Valc{\'{a}}zar.
\newblock {Infinitesimal symmetries in contact Hamiltonian systems}.
\newblock {\em J. Geom. Phys.} {\bf 153}:153651, 2020.
(\url{https://doi.org/10.1016/j.geomphys.2020.103651}).

\bibitem{LeSaVi2016}
M.~de~Le\'on, M.~Salgado, and S.~Vilari\~no.
\newblock {\em Methods of differential geometry in classical field theories:
  $k$-symplectic and $k$-cosymplectic approaches}.
\newblock World Scientific, 2016.

\bibitem{LS-2016}
M.~de~Le{\'{o}}n and C.~Sard{\'{o}}n.
\newblock {Cosymplectic and contact structures to resolve time-dependent and
  dissipative Hamiltonian systems}.
\newblock {\em J. Phys. A: Math. Theor.} {\bf 50}(25):255205, 2017.
(\url{https://doi.org/10.1088/1751-8121/aa711d}).

\bibitem{Galley-2013}
C.~R. Galley.
\newblock Classical mechanics of nonconservative systems.
\newblock {\em Phys. Rev. Lett.}, {\bf 110}(17):174301, 2013.
(\url{https://doi.org/10.1103/PhysRevLett.110.174301}).

\bibitem{GGMRR-2019}
J.~{Gaset}, X.~{Gr\`acia}, M.~{Mu\~noz-Lecanda}, X.~{Rivas}, and
  N.~{Rom\'an-Roy}.
\newblock {A contact geometry framework for field theories with dissipation}.
\newblock {\em Ann. Phys.}, {\bf 414}:168092, 2020.
(\url{https://doi.org/10.1016/j.aop.2020.168092}).

\bibitem{GGMRR-2020}
J.~{Gaset}, X.~{Gr\`acia}, M.~{Mu\~noz-Lecanda}, X.~{Rivas}, and
  N.~{Rom\'an-Roy}.
\newblock {A $k$-contact Lagrangian formulation for nonconservative field theories}.
\newblock {\em arXiv:2002.10458 [math-ph]}, 2020.

\bibitem{Geiges-2008}
H.~Geiges.
\newblock {\em An Introduction to Contact Topology}.
\newblock Cambridge University Press, 2008.

\bibitem{GeGu2002}
B.~Georgieva and R.~Guenther.
\newblock First {N}oether-type theorem for the generalized variational
  principle of {H}erglotz.
\newblock {\em Topol. Methods Nonlinear Anal.}, {\bf 20}(2):261--273, 2002.

\bibitem{got79}
M.~J. Gotay and J.~M. Nester.
\newblock Presymplectic {L}agrangian systems {I}: the constraint algorithm and
  the equivalence theorem.
\newblock {\em Ann. Inst. Henri Poincaré}, {\bf 30}(2):129--142, 1979.

\bibitem{Goto-2016}
S. Goto.
\newblock Contact geometric descriptions of vector fields on dually flat spaces
  and their applications in electric circuit models and nonequilibrium
  statistical mechanics.
\newblock {\em J. Math. Phys.}, {\bf 57}(10):102702, 2016.
(\url{https://doi.org/10.1063/1.4964751}).

\bibitem{Gracia2000}
X.~Gr{\`{a}}cia.
\newblock {Fibre derivatives: some applications to singular {L}agrangians}.
\newblock {\em Rep. Math. Phys.}, {\bf 45}(1):67--84, 2000.
(\url{https://doi.org/10.1016/S0034-4877(00)88872-2}).

\bibitem{Gracia2002}
X.~Gr{\`{a}}cia and J.~M. Pons.
\newblock {Symmetries and infinitesimal symmetries of singular differential
  equations}.
\newblock {\em J. Phys. A Math. Gen.}, {\bf 35}(24):5059--5077, 2002.
(\url{https://doi.org/10.1088/0305-4470/35/24/306}).

\bibitem{He-1930}
G. Herglotz, Ber\"uhrungstransformationen, Lectures at the University of Gottingen, 1930.

\bibitem{Her-1985}
G. Herglotz, 
{\em Vorlesungen \"uber die Mechanik der Kontinua}. 
Teubner-Archiv zur Mathematik~{\bf 3};
Teubner, Leipzig 1985.

\bibitem{KA-2013}
A.~L. Kholodenko.
\newblock {\em Applications of Contact Geometry and Topology in Physics}.
\newblock World Scientific, 2013.

\bibitem{LL-2018}
M.~{Lainz Valc{\'a}zar} and M.~{de Le{\'o}n}.
\newblock {Contact Hamiltonian systems}.
\newblock {\em J. Math. Phys.}, {\bf 60}(10):102902, 2019.
(\url{https://doi.org/10.1063/1.5096475}).

\bibitem{lib87}
P.~Libermann and C.~Marle.
\newblock {\em Symplectic Geometry and Analytical Mechanics}.
\newblock D. Reidel, 1987.

\bibitem{LIU2018}
Q.~Liu, P.~J. Torres, and C.~Wang.
\newblock Contact hamiltonian dynamics: variational principles, invariants, completeness and periodic behavior.
\newblock {\em Ann. Phys.}, {\bf 395}:26 -- 44, 2018.
(\url{https://doi.org/10.1016/j.aop.2018.04.035}).

\bibitem{MPR-2018}
N.~E. Mart\'inez-P\'erez and C.~Ram\'irez.
\newblock On the {L}agrangian description of dissipative systems.
\newblock {\em J. Math. Phys.}, {\bf 59}(3):032904, 2018.
(\url{https://doi.org/10.1063/1.5004796}).

\bibitem{RMS-2017}
H.~Ramirez, B.~Maschke, and D.~Sbarbaro.
\newblock Partial stabilization of input-output contact systems on a {L}egendre
  submanifold.
\newblock {\em IEEE Trans. Automat. Control}, {\bf 62}(3):1431--1437, 2017.
(\url{https://doi.org/10.1109/TAC.2016.2572403}).

\bibitem{Ra-2006}
M.~Razavy.
\newblock {\em Classical and quantum dissipative systems}.
\newblock Imperial College Press, 2006.

\bibitem{RRSV2011}
A.~M. Rey, N.~Rom\'an-Roy, M.~Salgado, and S.~Vilariño.
\newblock On the $k$-symplectic, $k$-cosymplectic and multisymplectic
  formalisms of classical field theories.
\newblock {\em J. Geom. Mech.}, {\bf 3}(1):113--137, 2011.
(\url{https://doi.org/10.3934/jgm.2011.3.113}).

\bibitem{R2019}
N.~{Rom\'an-Roy}.
\newblock {A summary on symmetries and conserved quantities of autonomous
  Hamiltonian systems}.
\newblock {\em J. Geom. Mech.}, 2020.
(\url{https://doi.org/10.3934/jgm.2020009}).

\bibitem{Vi-2018}
M.~Visinescu.
\newblock Contact Hamiltonian systems and complete integrability.
\newblock {\em AIP Conference Proceedings}, {\bf 1916}(1):020002, 2017.

\end{thebibliography}

\end{document}